\documentclass[11pt]{article} 

\usepackage[utf8]{inputenc} 

\usepackage{geometry} 
\usepackage{graphicx, subfig} 

\usepackage{booktabs} 
\usepackage{array} 
\usepackage{paralist} 
\usepackage{verbatim} 
\usepackage{subfig} 

\usepackage{fancyhdr} 
\usepackage{sectsty}
\allsectionsfont{\sffamily\mdseries\upshape} 

\usepackage[nottoc,notlof,notlot]{tocbibind} 
\usepackage[titles,subfigure]{tocloft} 

\usepackage{cite}

\usepackage{amsmath, amsfonts, amsthm, amssymb, amscd}
\usepackage{indentfirst}

\usepackage[pdftex,bookmarksnumbered]{hyperref}

\newtheorem{theorem}{Theorem}

\title{A Note on Computational Complexity of Dou Shou Qi}
\author{Zhujun Zhang \thanks{E-mail: zhangzhujun1988@163.com} \\Water Bureau of Fengxian District, Shanghai}

\date{April 27, 2019} 

\providecommand{\keywords}[1]{\textbf{\textit{Index terms---}} #1}

\begin{document}
\maketitle

\begin{abstract}

Dou Shou Qi is a Chinese strategy board game for two players.
We use a EXPTIME-hardness framework to analyse computational complexity of the game. 
We construct all gadgets of the hardness framework.
In conclusion, we prove that Dou Shou Qi is EXPTIME-complete.

\end{abstract}


\section{Introduction}

Dou Shou Qi is a modern Chinese strategy board game for two players, and its origins are not entirely clear.
The game is also called Jungle, The Jungle Game, Children's Chess, Oriental Chess and Animal Chess.
Dou Shou Qi is played on a rectangular board consisting of 7$\times$9 squares, and pieces move on the squares as in Chess.
Figure \ref{doushouqi} illustrates initial configuration of the game.
The rules of Dou Shou Qi could be found in \cite{doushouqiwiki} and \cite{doushouqichessvariants}, and we just introduce basic rules of the game here.

\begin{figure}[htbp]
	\centering  
	\includegraphics[width=0.4 \linewidth]{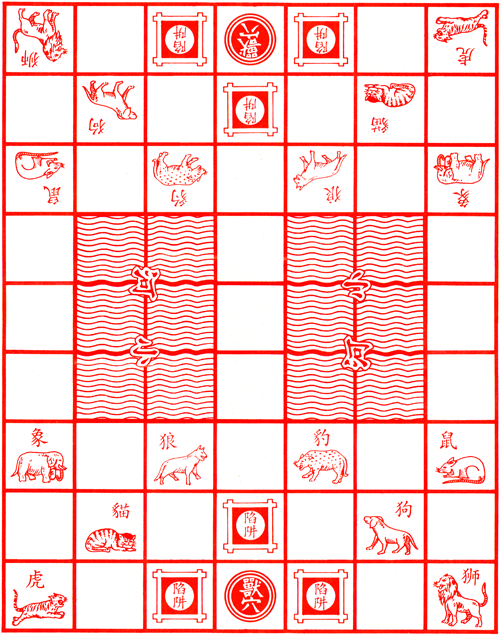}  
	\caption{Initial configuration of Dou Shou Qi.}  
	\label{doushouqi}   
\end{figure}

Each player controls eight different pieces representing different animals. 
Each animal has a certain rank, according to which they can capture opponent’s pieces. 
The pieces, from lower rank to higher rank, are:
Rat, Cat, Dog, Wolf, Leopard, Tiger, Lion and Elephant. 
Only pieces with equal or higher rank may capture an opponent’s piece. 
The only exception to this rule is that the weakest rat may capture the strongest elephant, just like the spy in Stratego.

In Dou Shou Qi, players alternate moves, and all pieces can move one square horizontally or vertically.
There are three special squares of the game board: Den, Trap and River.
The player who is first to move any one of their pieces into the opponent's den wins the game.
A piece that enters one of the opponent's trap squares is reduced in rank to 0.
Thus the trapped piece may be captured by the defending side with any piece, regardless of rank. 
A trapped piece has its normal rank restored when it exits an opponent's trap square.
A piece in one of its own traps is unaffected. 
The rat is the only piece that may move into a river square.
Only the lion and tiger may jump over a river horizontally or vertically. 
They jump from a square on one edge of the river to the next non-river square on the other side. 
If that square contains an enemy piece of equal or lower rank, the lion or tiger may capture it. 

In this note, we review related work in Section 2;
in Section 3, we describe notations and discuss complexity of Dou Shou Qi; we summarize this note in Section 4.

\section{Related Work}

In last years, people researched computational complexity and endgame database of Dou Shou Qi.
In 2010, Burnett \cite{doushouqisubproblem} analysed loosely coupled subproblems of Dou Shou Qi.
In 2012, van Rijn proved Dou Shou Qi to be PSPACE-hard by reduction from a game called Bounded 2CL introduced in \cite{constraintlogic}.
In 2013, van Rijn and Vis \cite{doushouqiendgame1} \cite{doushouqiendgame2} provided another PSPACE-hardness proof of Dou Shou Qi, and they established an endgame database containing all configurations with up to four pieces.
In 2016, Bohrweg \cite{doushouqiendgame3} computed the seven-piece endgame database of Dou Shou Qi with a parallel variation of retrograde analysis. 

In last decades, many classic two-player board games were proved to be computationally hard.
In 1980, Lichtenstein and Sipser \cite{gopspace} proved Go with the superko rule to be PSPACE-hard, and Reisch \cite{gomoku} proved Gomoku (Gobang) to be PSPACE-complete.
In 1981, Fraenkel and Lichtenstein \cite{chess} proved Chess to be EXPTIME-complete.
In 1983, Robson \cite{goexptime} proved  Go with the basic ko rule to be EXPTIME-complete.
In 1984, Robson \cite{checkers} proved Checkers to be EXPTIME-complete.
In 1987, Adachi et al. \cite{shogi} proved Shogi to be EXPTIME-complete.
In 1994, Iwata and Kasai \cite{othello} proved Othello (Reversi) to be PSPACE-complete.
Demaine and Hearn review results
about the complexity of many classic games in their survey paper \cite{playinggameswithalgorithms}.
Other results about complexity of well-known games could be found in Wiki page ``Game complexity'' \cite{gamecomplexity}.

More recently, researchers focus on frameworks which could be used to prove hardness of games.
In 2010, Fori\v{s}ek \cite{2dplatformgames} introduced a basic NP-hardness framework for 2D platform games.
In 2012, Aloupis, Demaine and Guo \cite{nintendogamesnphard} introduced a elegant NP-hardness framework, and they used the framework to analyse complexity of several video games.
In 2014, Viglietta \cite{gamehardjob} established some general schemes relating the computational complexity of 2D platform games. 
In 2015, Aloupis et al. \cite{nintendogameshard} introduced a new  PSPACE-hardness framework.
In our note \cite{mynotehardnessframework}, we introduced a EXPTIME-hardness framework for 2D games, and we will use this framework to discuss complexity of Dou Shou Qi in this note.

\section{Complexity of Dou Shou Qi}

Before we discuss complexity of Dou Shou Qi, we introduce the notations used in this note.
Pieces and special squares table of the game is illustrated in Figure \ref{piecestable}.
We use ``Red'' and ``Black'' to represent two players in the game.
We use two capital letters to represent one piece in this note, for example, ``RE'' means ``red elephant''.
In this note, we just need four types of pieces (elephant, lion, tiger and cat) to establish reduction.
We use different colour squares to represent dens, traps and rivers in Dou Shou Qi.

\begin{figure}[htbp]
	\centering  
	\includegraphics[width=0.65 \linewidth]{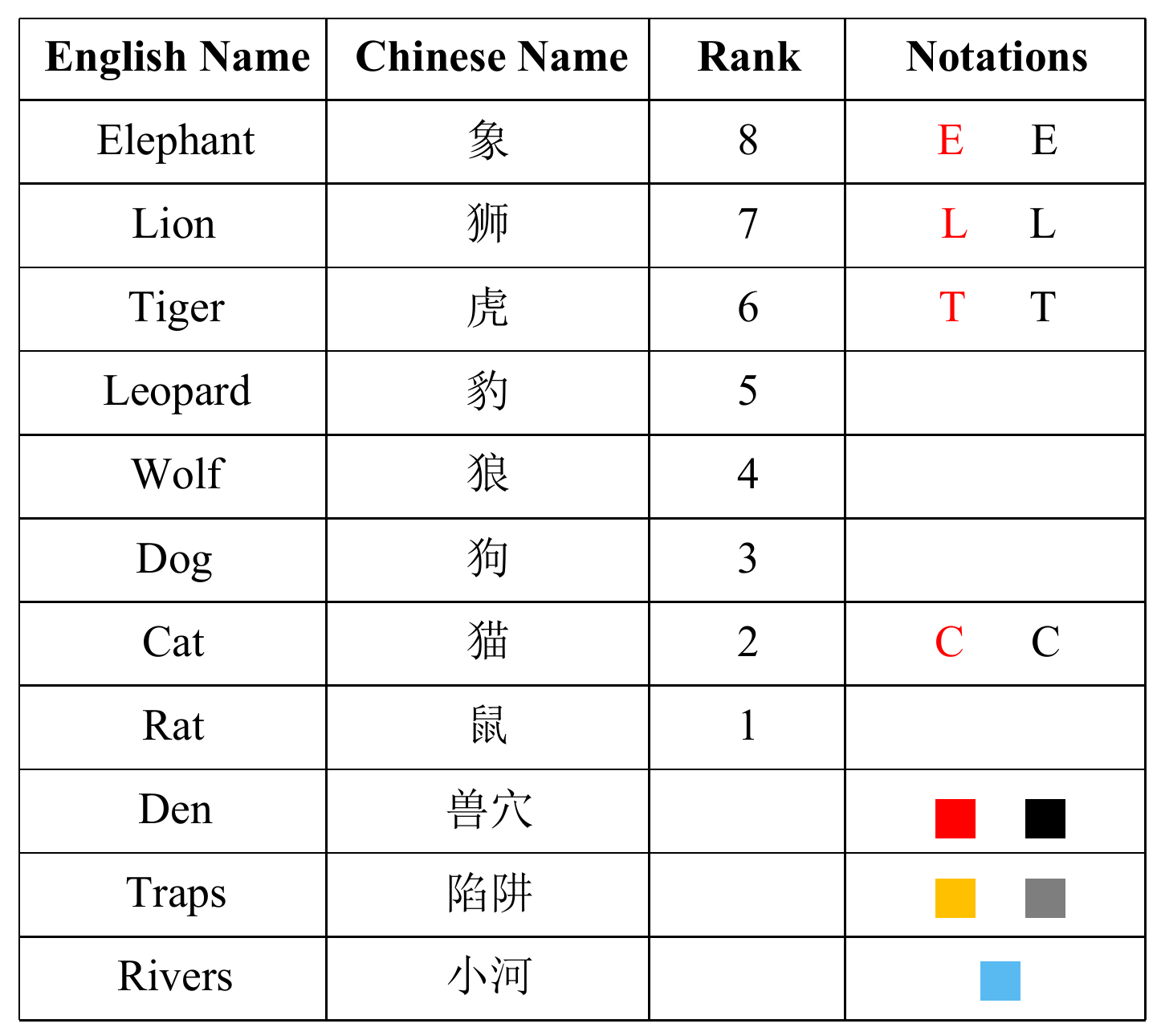}  
	\caption{Pieces and special squares table of Dou Shou Qi.}  
	\label{piecestable}   
\end{figure}

Decision problem of Dou Shou Qi is to decide whether Red has a forced win in a given position.
There are eight pieces for each player, and there are three special squares, thus there are at most $(8 \times 2 + 1) \times (3 + 1) = 68$ states for each square.
The number of all configuration in a $n \times n$ game board is bounded by $68 ^ {n^2}$, therefore Dou Shou Qi is in EXPTIME.
To prove EXPTIME-hardness of Dou Shou Qi, we need to construct all gadgets of the EXPTIME-hardness framework introduced in our note \cite{mynotehardnessframework}.

In our reduction, each player control one tiger respectively, and these two tigers indicate the avatars in the hardness framework.
The start, finish, turn, switch, merge and crossover gadgets are trivial.
Our one-way gadget is inspired by van Rijn and Vis \cite{doushouqiendgame1}.
Since the BBB, BRB and RRB door gedgets could be constructed symmetrically, we just need to construct RRR, RBR and BBR door gadgets for Dou Shou Qi.

In all these gadgets for Dou Shou Qi, paths for the avatars are composed of river squares so that only tigers and lions may enter and traverse these paths.
On the other hand, the tigers can not leave the paths, since the paths are surrounded by a large number of unmoveable elephants.
The border between unmoveable red elephants and black elephants are river squares, thus these elephants can not attack each other.

\textbf{Start, finish, turn, switch and merge gadgets.} 
It is easy to construct start, finish, turn, switch and merge gadgets for Dou Shou Qi, and these gadgets are illustrated in Figure \ref{start} \ref{finish} \ref{turn} \ref{switch} respectively. 
In the start gadget, the RT can move east to leave the gadget.
Once the RT enters the finish gadget from west, it can arrive black den, and Red wins immediately. 
The RT can enter a turn gadget from north, and it can move east to leave. 
When the RT enters a switch gadget from south, it can move east or west to leave. 
Moreover, a merge gadget is identical to a switch gadget in Dou Shou Qi.
The start, finish, turn, switch and merge gadgets for Black could be constructed symmetrically.

\textbf{Crossover gadget.} 
Three types of crossover gadgets are identical in Dou Shou Qi, and Figure \ref{crossover} illustrates a crossover gadget. 
The RT and BT can traverse a crossover gadget through two paths, and there is no leakage between two paths of the gadget.
For one path, the RT (BT) enters the gadget from west, and then it moves east and jumps over the rivers to leave. 
Situation of another path is similar. 

\textbf{One-way gadget.} 
Figure \ref{oneway} illustrates an one-way gadget for Dou Shou Qi. 
In this gadget, only one BC is moveable. 
The RT can only traverse the gadget from west to east. 
When the RT enters an one-way gadget from west, the BC has to move south or north to avoid being captured.
Then the RT can move east to traverse the gadget, and the BC should move back to the square where it starts from when the RT leaves the gadget.
On the other hand, if the RT enters an one-way gadget from east, it can not traverse the gadget, since it will be captured by the BC at one of two black traps (gray squares).
Moreover, the RT can never capture the BC in an one-way gadget, and the BC can not leave the gadget since only Lion and Tiger may jump over rivers in Dou Shou Qi.
The one-way gadget for Black also could be constructed symmetrically.

\begin{figure}
	\begin{minipage}[htbp]{0.5\linewidth}
		\centering
		\includegraphics[width=0.5 \linewidth]{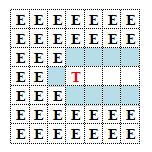}
		\caption{Start gadget.}
		\label{start}
	\end{minipage}%
	\begin{minipage}[htbp]{0.5\linewidth}
		\centering
		\includegraphics[width=0.5 \linewidth]{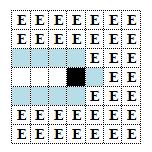}
		\caption{Finish gadget.}
		\label{finish}
	\end{minipage}
\end{figure}

\begin{figure}
	\begin{minipage}[t]{0.5\linewidth}
		\centering
		\includegraphics[width=0.5 \linewidth]{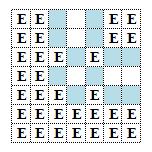}
		\caption{Turn gadget.}
		\label{turn}
	\end{minipage}%
	\begin{minipage}[t]{0.5\linewidth}
		\centering
		\includegraphics[width=0.5 \linewidth]{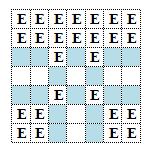}
		\caption{Switch and merge gadgets.}
		\label{switch}
	\end{minipage}
\end{figure}

\begin{figure}
	\begin{minipage}[t]{0.5\linewidth}
		\centering
		\includegraphics[width=0.5 \linewidth]{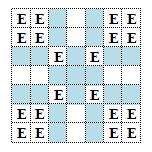}
		\caption{Crossover gadget.}
		\label{crossover}
	\end{minipage}
	\begin{minipage}[t]{0.5\linewidth}
		\centering
		\includegraphics[width=0.68 \linewidth]{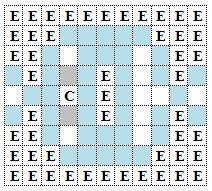}
		\caption{One-way gadget.}
		\label{oneway}
	\end{minipage}%
\end{figure}

\textbf{RRR door gadget.} 
Figure \ref{RRRdoor} illustrates a RRR door gadget for Dou Shou Qi. 
In this gadget, only five RE's (at r11, c11, h6, p6 and x11), eight RC's (at i9, i10, i12, i13, o9, o10, o12 and o13) and three BL's (at f11, l5 and r14) are moveable. 
We call the RE at r11 control RE.
When control RE stops at r11, the gadget is considered in the closed state; when control RE stops at f11, the gadget is considered in the open state.

\begin{figure}[htbp]
	\centering  
	\includegraphics[width=0.9 \linewidth]{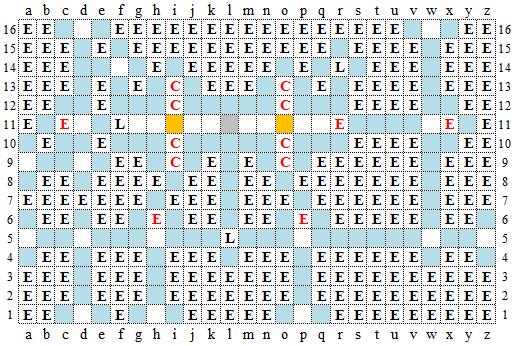}  
	\caption{RRR door gadget.}  
	\label{RRRdoor}   
\end{figure}

Squares a5 and d1 are the entrance and the exit of the open path respectively. 
Path of square sequence (a5, d5, d1) is an open path. 
Squares a9 and d16 are the entrance and the exit of the traverse path respectively. 
Path of square sequence (a9, d9, d11, d16) is a traverse path. 
Squares z5 and w16 are the entrance and the exit of the close path respectively. 
Path of square sequence (z5, w5, w11, w16) is a close path.
Squares h1 and p1 are entrances of rapid checkmate paths for Red. 

The RT can open the gadget by traversing the open path. 
When the RT moves to d5, control RE can move to f11, which forces the BL at f11 move to f14.
If the BL at l5 captures control RE at l11 (black trap square), the RT can enter the rapid checkmate path via h5.
Once control RE moves to f11, the gadget is in the open state, since the RT can traverse the traverse path safely.
Moreover, the BL at r14 should move to r11 in order to protect square w11 when control RE moves west.

The RT can traverse the traverse path if and only if the gadget is in the open state.
If control RE stops at f11, the BL can not protect square d11, thus the RT can traverse the path. 
If the gadget is in the closed state, the RT will be captured by the BL at f11 when it traverses the traverse path and moves to d11.

When the RT traverses the close path, it has to close the gadget.
To traverse the path, the RT must pass square w11, however, the square is protected by the BL at r11. 
Thus, after the RT moves to w5, control RE must move to r11 to drive the BL away, which makes the gadget in the closed state.
Then the RT can traverse the close path safely.
When control RE moves to l11, the BL at l5 can not capture control RE, since the RT can enter rapid checkmate path via p5.
Moreover, the BL at f14 should move to f11 in order to protect square d11 when control RE moves east.

Notice, the RE's and RC's can not leave the gadget, since they may not jump over rivers.
If the BL's attempt to leave the gadget, they will be captured by the RE's immediately.
If the RT is not at the RRR door gadget, control RE can not move to l11 (black trap square), since it will be captured by the BL at l5.
Three BL's can not move to i11 or o11 (red trap squares), since these two squares are protected by eight RC's.
Moreover, if the RT moves to f11 or r11, one of the BL's will capture it.

\textbf{RBR door gadget.}
Figure \ref{RBRdoor} illustrates a RBR door gadget for Dou Shou Qi. 
In this gadget, only four RE's (at f11, h6, p6 and x11), eight RC's (at i9, i10, i12, i13, o9, o10, o12 and o13) and two BL's (at l5 and r11) are moveable. 
We call the RE at f11 control RE.
When control RE stops at f11, the gadget is considered in the closed state; when control RE stops at r11, the gadget is considered in the open state.

Squares z5 and w16 are the entrance and the exit of the open path respectively. 
Path of square sequence (z5, w5, w11, w16) is an open path for Red.
Squares a11 and f16 are the entrance and the exit of the traverse path respectively. 
Path of square sequence (a11, f11, f16) is a traverse path for Black. 
Squares a5 and d1 are the entrance and the exit of the close path respectively. 
Path of square sequence (a5, d5, d1) is a close path for Red.
Squares h1 and p1 are entrances of rapid checkmate paths for Red. 

The open path in a RBR door gadget is identical to the close path in a RRR door gadget.
When the RT traverses the open path, it has to open the RBR door gadget.

The BT can traverse the traverse path if and only if the gadget is in the open state.
If control RE stops at f11, the RT can not traverse the path, since tigers can not capture elephants in Dou Shou Qi except in a trap square.

The close path in a RBR door gadget is identical to the open path in a RRR door gadget.
The RT can close the RBR door gadget by traversing the close path. 

Notice, the RE's and RC's can not leave the gadget.
If the BL's attempt to leave the gadget, they will be captured by the RE's.
If the RT is not at the RRR door gadget, control RE can not move to l11 (black trap square), since it will be captured by the BL at l5.
Two BL's and the BT can not move to i11 or o11 (red trap squares), since these two squares are protected by eight RC's.
Moreover, if the RT moves to r11, the BL at r14 will capture it.

\begin{figure}[htbp]
	\centering  
	\includegraphics[width=0.9 \linewidth]{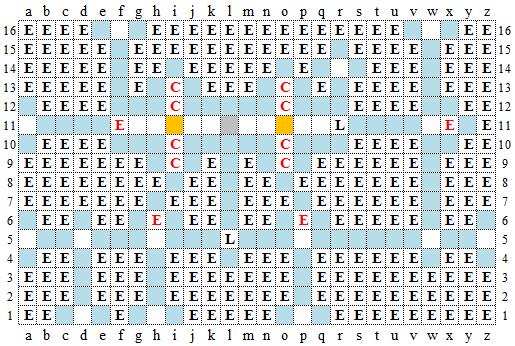}  
	\caption{RBR door gadget.}  
	\label{RBRdoor}   
\end{figure}

\begin{figure}[htbp]
	\centering  
	\includegraphics[width=0.9 \linewidth]{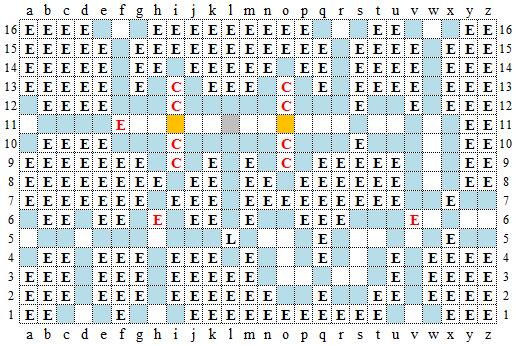}  
	\caption{BBR door gadget.}  
	\label{BBRdoor}   
\end{figure}

\textbf{BBR door gadget.} 
Figure \ref{BBRdoor} illustrates a RBR door gadget for Dou Shou Qi. 
In this gadget, only three RE's (at f11, h6 and v6), eight RC's (at i9, i10, i12, i13, o9, o10, o12 and o13) and one BL (at l5) are moveable. 
We call the RE at f11 control RE.
When control RE stops at f11, the gadget is considered in the closed state; when control RE stops at r11, the gadget is considered in the open state.

Squares z6 and w16 are the entrance and the exit of the open path respectively. 
Path of square sequence (z6, y6, w6, w8, w9, w10, w11, w16) is an open path for Black.
Squares a11 and f16 are the entrance and the exit of the traverse path respectively. 
Path of square sequence (a11, f11, f16) is a traverse path for Black. 
Squares a5 and d1 are the entrance and the exit of the close path respectively. 
Path of square sequence (a5, d5, d1) is a close path for Red.
Square h1 is entrance of rapid checkmate path for Red, and squares r1 and v1 are entrances of rapid checkmate paths for Black.

The BT can open the gadget by traversing the open path.
When the BT moves to y6, it can not move to w6 directly, since the RE at v6 protect the square.
The BL at l5 should move to t5 in order to help the BT pass w6.
If the RE captures the BT at w6, the BL can enter rapid checkmate path via v5.
After the BL leaves l5, control RE has to move to r11, otherwise, the BT can enter rapid checkmate path via r11 in certain steps.
Notice, control RE has sufficient time to move to r11 before the BT arrives r11.
When the BT moves to w11, the BL should move back to l5 in order to prevent control RE from moving west.

The traverse path and close path in a BBR door gadget are identical to the paths in a RBR door gadget.
The BT can traverse the traverse path if and only if the gadget is in the open state.
The RT can close the gadget by traversing the close path. 

Notice, the RE's and RC's can not leave the gadget.
If the BL attempts to leave the gadget, it will be captured by the RE's immediately.
If the RT is not at the RRR door gadget, control RE can not move to l11 (black trap square), since it will be captured by the BL at l5.
The BL and BT can not move to i11 or o11 (red trap squares), since these two squares are protected by eight RC's.
Moreover, if the RT stops at d5, the BT can neither traverse the open path nor open the gadget.
However, two different colour avatars will never enter a BBR door gadget at the same time in the EXPTIME-hardness framework.

As all gadgets of EXPTIME-hardness framework have been constructed in Dou Shou Qi, we obtain the following result.

\begin{theorem}
	Dou Shou Qi is EXPTIME-complete.
\end{theorem}

\begin{proof}
	Firstly, Dou Shou Qi could be solved by a brute search algorithm in exponential time, thus Dou Shou Qi is in EXPTIME. 
	Secondly, we have constructed all gadgets of the EXPTIME-hardness framework in Dou Shou Qi, thus Dou Shou Qi is EXPTIME-hard.
\end{proof}

\textbf{Remark.}
We take the liberty to freely use river squares and traps in the gadgets.
However, the original game board contains several properties, such as clustered river squares and only traps around the dens.
Moreover, We do not discuss the reachability of the instance constructed here.
In fact, it seems that the above instance of Dou Shou Qi will never occur in reality.

\section{Conclusion}

We use a EXPTIME-hardness framework to analyse computational complexity of Dou Shou Qi. 
We construct all gadgets of the hardness frameworks in Dou Shou Qi.
In conclusion, we prove that Dou Shou Qi is EXPTIME-complete.
In future, We will try to use the hardness frameworks to discuss complexity of other interesting board games, such as Stratego and Arimaa.


\bibliographystyle{plain}
\bibliography{ref}

\end{document}